\documentclass[aps,manuscript,pop,showpacs,preprint,superscriptaddress]{revtex4-1}
\usepackage{lmodern}

\usepackage{lmodern}
\usepackage[T1]{fontenc}
\usepackage[latin9]{inputenc}
\usepackage{prettyref}
\usepackage{amsmath}
\usepackage{amsthm}
\usepackage{mathdots}
\usepackage{graphicx}

\makeatletter
\theoremstyle{plain}
\newtheorem{thm}{\protect\theoremname}
\ifx\proof\undefined
\newenvironment{proof}[1][\protect\proofname]{\par
\normalfont\topsep6\p@\@plus6\p@\relax
\trivlist
\itemindent\parindent
\item[\hskip\labelsep\scshape #1]\ignorespaces
}{%
\endtrivlist\@endpefalse
}
\providecommand{\proofname}{Proof}
\fi
\theoremstyle{plain}
\newtheorem{cor}[thm]{\protect\corollaryname}

\usepackage{amsthm}\usepackage{latexsym}\usepackage{bm}\usepackage{amsfonts}
\makeatother

\providecommand{\corollaryname}{Corollary}
\providecommand{\theoremname}{Theorem}
\begin{document}

\title{What breaks parity-time-symmetry? \textemdash{} pseudo-Hermiticity
and resonance between positive- and negative-action modes }

\author{Ruili Zhang}

\affiliation{School of Science, Beijing Jiaotong University, Beijing 100044, China}

\author{Hong Qin }

\thanks{Corresponding author, hongqin@ustc.edu.cn}

\affiliation{School of Physics, University of Science and Technology of China, Hefei, Anhui 230026,
China}

\affiliation{Plasma Physics Laboratory, Princeton University, Princeton, NJ 08543,
USA}

\author{Jianyuan Xiao}

\affiliation{School of Physics, University of Science and Technology of China, Hefei, Anhui 230026,
China}

\author{Jian Liu}

\affiliation{School of Physics, University of Science and Technology of China, Hefei, Anhui 230026,
China}
\begin{abstract}
It is generally believed that Parity-Time (PT)-symmetry breaking occurs
when eigenvalues or both eigenvalues and eigenvectors coincide. However,
we show that this well-accepted picture of PT-symmetry breaking is
incorrect. Instead, we demonstrate that the physical mechanism of
PT-symmetry breaking is the resonance between positive- and negative-action
modes. It is proved that PT-symmetry breaking occurs when and only
when this resonance condition is satisfied, and this mechanism applies
to all known PT-symmetry breakings observed in different branches
of physics. The result is achieved by proving a remarkable fact that
in finite dimensions, a PT-symmetric Hamiltonian is necessarily pseudo-Hermitian,
regardless whether it is diagonalizable or not.
\end{abstract}

\maketitle

It is a fundamental assumption of quantum physics that observables
are Hermitian operators. Bender and collaborators \cite{bender1998real,bender2002complex,bender2007making}
proposed to relax this assumption by considering Parity-Time (PT)-symmetric
operators. Since its conception, PT-symmetry has been found important
applications in many branches of physics \cite{jones1999energy,mostafazadeh2002pseudo,heiss2004exceptional,bender2007making,el2007theory,makris2008beam,makris2010pt,schomerus2010quantum,chong2011p,ge2012conservation,ramezani2012exceptional},
including classical physics \cite{klaiman2008visualization,schindler2011experimental,peng2014parity,hodaei2017enhanced}
and quantum physics \cite{dorey2007ode,musslimani2008optical,longhi2009bloch,lin2011unidirectional,szameit2011p,regensburger2012parity,sarma2014continuous,ablowitz2016inverse,zhang2016observation,jahromi2017statistical}.
PT-symmetry is also observed in many laboratory experiments \cite{guo2009observation,ruter2010observation,feng2011nonreciprocal,bittner2012p,liertzer2012pump}.

Among all the current research topics actively pursued in the field,
the breaking of PT-symmetry is arguably the most prominent one. The
question needs to be answered is when and how PT-symmetry breaking
occurs as system parameters vary. In this aspect, the role of exceptional
points (EPs), which are defined by Kato \cite{kato1966perturbation}
to be points in the parameter space where two or more eigenvalues
of the system coincide, has been noticed. In most of the existing
literature, PT-symmetry
breaking is described as the process of two real eigenvalues of the
Hamiltonian coinciding on the real axis at EPs and then moving off
the real axis in opposite direction to generate a complex conjugate
pair of eigenvalues. It is often stated or implied
that EPs are the locations for PT-symmetry breaking \cite{brandstetter2014reversing,heiss2004exceptional,heiss2012physics}.
Recently, some researchers \cite{klaiman2008visualization,ge2016anomalous,feng2017non,Ashida2017PT,el2018non}
emphasized that in order to break the PT-symmetry at EPs, the eigenvectors
need to coalesce as well. This condition of coalescence of both eigenvalues
and eigenvectors is exactly equivalent to the condition that the Hamiltonian
$H$ is not diagonalizable at the EPs.

In this paper, we show that these well-known pictures of PT-symmetry
breaking are incorrect. To motivate the discussion, we use the following
example to show that (i) PT-symmetry breaking occurs at some EPs,
but does not at all EPs, and (ii) when PT-symmetry breaking occurs
at EPs, the Hamiltonian can be either non-diagonalizable (coalescence
of both eigenvalues and eigenvectors) or diagonalizable (coalescence
of eigenvalues only). Let's consider the Hamiltonian
\begin{equation}
H=\left(\begin{array}{cccc}
-3 & c & 0 & 0\\
c & -3 & 0 & 0\\
b-ci & 7i+a & 4-ia & ib\\
-7i+a & b+ic & -ib & 4+ia
\end{array}\right)\thinspace,\label{eq:Ex1}
\end{equation}
where $a,b$ and $c$ are real numbers. For the parity transformation
that switches the first row with the second row, and the third row
with the fourth row, $H$ is PT-symmetric. Its eigenvalues are $(-3\pm c,4\pm\sqrt{-a^{2}+b^{2}})$.
In the parameter space, EPs locate at $c=0$ and $a=\pm b$, and the
corresponding coincident eigenvalues are $-3$ and $4$, respectively.
Obviously, PT-symmetry breaking occurs at the EPs with $a=\pm b$.
Two of the eigenvalues $(4\pm\sqrt{-a^{2}+b^{2}})$ move towards each
other when $b^{2}$ approaches $a^{2}$ from the right on the real
axis, and become a complex conjugate pair when $b^{2}$ is less than
$a^{2}.$ On the other hand, there is no PT-symmetry breaking at the
EPs with $c=0,$ i.e., the eigenvalues $(-3\pm c)$ are always real
in the neighborhood of $c=0$ on the real axis.

This example also shows that non-diagonalizable Hamiltonian at EPs
(coalescence of both eigenvalues and eigenvectors) is a sufficient
but not necessary condition for PT-symmetry breaking. In Fig.\,\ref{f0-1},
the parameter space for $a$ and\textbf{ $b$} is plotted for a fixed
$c.$ When $(a,b)$ are in the shaded regions, the eigenvalues $(4\pm\sqrt{-a^{2}+b^{2}})$
are a complex conjugate pair, and when $(a,b)$ are in the un-shaded
regions, the eigenvalues $(4\pm\sqrt{-a^{2}+b^{2}})$ are real. The
points on the two lines $a=\pm b$ are EPs. It is easy to verify that
at the EPs defined by $a=\pm b\neq0$, $H$ is not diagonalizable.
But at the EP $(a=0,b=0)$, $H$ is diagonalizable. When the parameters
vary along Path 1 (blue line), PT-symmetry breaking occurs at $EP_{1}$,
where $H$ is not diagonalizable. But when the parameters vary along
Path 2 (red curve), PT-symmetry breaking occurs at $EP_{2}$, where
$H$ is diagonalizable. This shows that PT-symmetry breaking can happen
at EPs where the Hamiltonian can be either diagonalizable or non-diagonalizable.

This example will be analyzed in details as Example 1 after the development
of our general theory.

\begin{figure}
\includegraphics[scale=0.55]{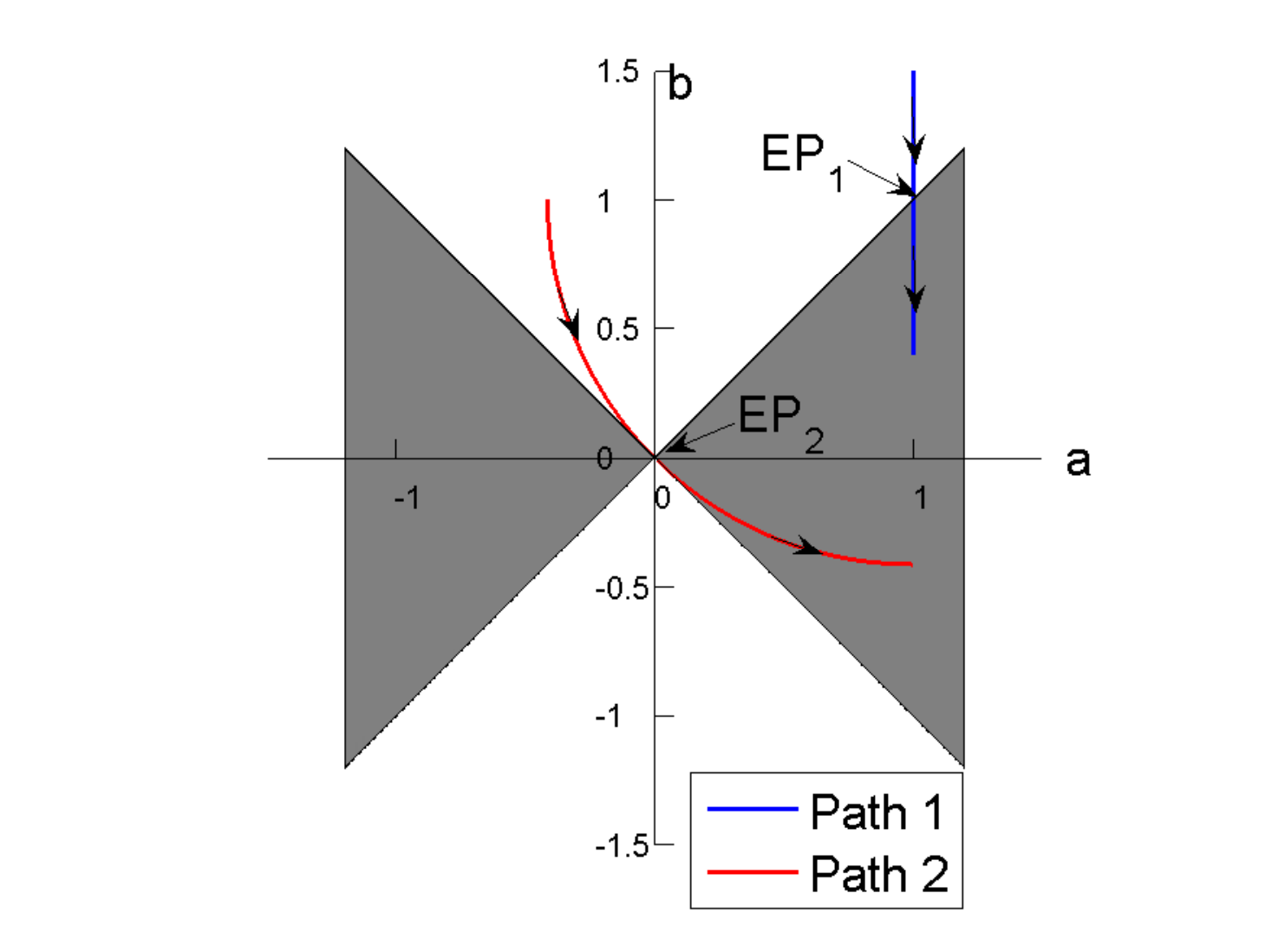}

\caption{Parameter space for eigenvalues $(4\pm\sqrt{-a^{2}+b^{2}})$ of $H$.
In the shaded regions, the eigenvalues are a complex conjugate pair.
When the parameters vary along Path 1 (blue line), PT-symmetry breaking
occurs at $EP_{1}$, where $H$ is not diagonalizable. But when the
parameters vary along Path 2 (red line), there is no PT-symmetry breaking
at $EP_{2}$, where $H$ is diagonalizable.}
\label{f0-1}
\end{figure}

In general, for a finite dimensional system, EP is a necessary but
not sufficient condition for PT-symmetry breaking, and EP with coalescence
of both eigenvalues and eigenvectors is a sufficient but not a necessary
condition. What is a necessary and sufficient condition for PT-symmetry
breaking?

In this paper, we will give such a condition. We will first prove
a remarkable fact that all PT-symmetric Hamiltonians in the finite
dimensions are pseudo-Hermitian \cite{Lee69,mostafazadeh2002pseudo,mostafazadeh2002pseudoII,mostafazadeh2002pseudoIII}.
(See Eq.\,\eqref{eq:PH} for definition). As a result, an action
can be assigned to each eigenmode of a PT-symmetric system. It is
shown that if an EP is caused by the resonance between two eigenmodes
with the same sign of action, then there is no PT-symmetry breaking
at this EP. Otherwise, i.e., if a positive-action mode resonates with
a negative-action mode at an EP, then PT-symmetry breaking will occur
along a curve passing through this EP in the parameter space. Therefore,
resonance between positive- and negative-action modes is a necessary
and sufficient condition and thus the physical mechanism for PT-symmetry
breaking.

Our result is built upon the mathematical work on the stability of
G-Hamiltonian system by Krein, Gel'fand and Lidskii \cite{Krein1950,Gel1955,KGML1958}
in 1950s. It turns out that the definitions of G-Hamiltonian \cite{Krein1950,Gel1955,KGML1958}
and pseudo-Hermitian \cite{Lee69,mostafazadeh2002pseudo,mostafazadeh2002pseudoII,mostafazadeh2002pseudoIII}
are identical for finite dimensions. It is probably more appropriate
to adopt the jargon of ``G-Hamiltonian'', since it appeared earlier
than ``pseudo-Hermitian''. However, to be more accessible to the
physics community, we will use both. By proving the fact that all
finite dimensional PT-symmetric systems are pseudo-Hermitian (or G-Hamiltonian),
we are able to borrow the results established by Krein, Gel'fand
and Lidskii. Especially, the action we defined for the eigenmodes
is the same as the Krein signature for the eigenmodes of G-Hamiltonian
systems; and the resonance between positive- and negative-action modes
is the same as the well-known Krein collision.

There is another fact that needs to be pointed out. Shortly after
the concept of PT-symmetry introduced by Bender et al. \cite{bender1998real},
Mostafazadeh proved that a diagonalizable PT-symmetric Hamiltonian
is pseudo-Hermitian \cite{mostafazadeh2002pseudo,mostafazadeh2002pseudoII,mostafazadeh2002pseudoIII}.
We note that Mostafazadeh's result is different from ours, which states
that a finite dimensional PT-symmetric Hamiltonian is always pseudo-Hermitian,
whether it is diagonalizable or not. The difference is significant,
because, as we pointed out previously, PT-symmetry breaking occurs
at EPs where the Hamiltonian can be either diagonalizable or non-diagonalizable.
Apparently, our result is more suitable for the investigation of the
the mechanism of PT-symmetry breaking.

We start our investigation from the definitions of PT-symmetry and
pseudo-Hermiticity (or G-Hamiltonian property) for the linear system
specified by a Hamiltonian $H$,
\begin{equation}
\dot{\boldsymbol{x}}=-iH\boldsymbol{x}=A\boldsymbol{x}\thinspace,\label{eq:ls}
\end{equation}
where $A$ is defined to be a shorthand notation of $-iH.$

The Hamiltonian $H$ is called PT-symmetric if
\begin{equation}
PTH-HPT=0\thinspace,\label{eq:PT}
\end{equation}
where $P$ is a linear operator satisfying $P^{2}=I$ and $T$ is
the complex conjugate operator \cite{bender2007making}. In finite
dimensions, which will be the focus of the present study, $H$, $A$
and $P$ can be represented by matrices, and the PT-symmetric condition
Eq.\,\eqref{eq:PT} is equivalent to
\begin{equation}
P\bar{H}-HP=0\textrm{ \thinspace\thinspace or\thinspace\thinspace\ }PA+\bar{A}P=0\thinspace,\label{eq:PT1}
\end{equation}
Here, $\bar{H}$ and $\bar{A}$ denote the complex conjugates of $H$
and $A$, respectively. The PT-symmetry condition can be understood
as follows. In terms of a $P$-reflected variable $\boldsymbol{y}\equiv P\boldsymbol{x}$,
Eq.\,\eqref{eq:ls} is
\begin{equation}
\dot{\boldsymbol{y}}=PAP^{-1}\boldsymbol{y}\thinspace.\label{eq:y}
\end{equation}
In general Eq.\,\eqref{eq:ls} is not $P$-symmetric, i.e., $PAP^{-1}\neq A$.
However, we can check the effect of applying an additional $T$-reflection,
i.e. $t\rightarrow-t$ and $i\rightarrow-i$. Then Eq.\,\eqref{eq:y}
becomes
\begin{equation}
\dot{\bar{\boldsymbol{y}}}=-\overline{PAP^{-1}}\bar{\boldsymbol{y}}\thinspace.
\end{equation}
If $-\overline{PAP^{-1}}=A$, which is equivalent to PT-symmetric
condition Eq.\,\eqref{eq:PT} or Eq.\,\eqref{eq:PT1}, then Eq.\,\eqref{eq:y}
in terms of $\bar{\boldsymbol{y}}$ is identical to Eq.\,\eqref{eq:ls}
in terms of $\boldsymbol{x}.$ Thus, PT-symmetry is an invariant property
of the system under the reflections of both parity and time.

We next define pseudo-Hermiticity or G-Hamiltonian property for the
finite-dimensional linear system Eq.\,\eqref{eq:ls}. A Hamiltonian
$H$ is called pseudo-Hermitian if there exist a non-singular Hermitian
matrix $G$ and Hermitian matrix $S$ such that
\begin{equation}
H=-G^{-1}S\thinspace;\label{eq:PH}
\end{equation}
or equivalently if there exist a non-singular Hermitian matrix $G$
and Hermitian matrix $S$ such that such that $A$ satisfying
\begin{equation}
A=iG^{-1}S\thinspace.\label{eq:GH0}
\end{equation}
A matrix that can be decomposed as in Eq.\,\eqref{eq:GH0} is called
G-Hamiltonian \cite{KGML1958}. It is easy to verify that $A$ is
G-Hamiltonian or $H$ is pseudo-Hermitian if and only if there exists
a non-singular Hermitian matrix $G$ such that
\begin{equation}
A^{*}G+GA=0\thinspace,\label{eq:GH}
\end{equation}
where $A^{*}$ is the conjugate transpose of the matrix $A$. As
mentioned previously, the concept of G-Hamiltonian was introduced
by Krein, Gel'fand and Lidskii \cite{Krein1950,Gel1955,KGML1958}
in 1950s, and pseudo-Hermiticity was introduced by Lee and Wick independently
in 1969 \cite{Lee69}. For finite dimensional systems, the condition
for $H$ to be pseudo-Hermitian is identical to that for $A$ to be
G-Hamiltonian. We will use both terminologies. Note that the eigenvalues
of $A$ are related to those of $H$ by a simple factor of $-i.$

PT-symmetry and pseudo-Hermiticity (or G-Hamiltonian property) are
two important geometric or physical properties of the system under
investigation. We now establish a connection between PT-symmetry and
pseudo-Hermiticity in finite dimensions. First, we give the following
necessary and sufficient condition for a matrix to be G-Hamiltonian.

\begin{thm}
For a matrix $A\in C^{n\times n}$, it is G-Hamiltonian if and if
only it is similar to $-\bar{A}$, where $\bar{A}$ is its complex
conjugate .
\end{thm}

\begin{proof}
Necessity is easy to prove. If a matrix is G-Hamiltonian, i.e. satisfying
Eq.\,\eqref{eq:GH}, then $A=-G^{-1}A^{*}G$. Thus matrix $A$ is
similar to $-A^{*}$, and also to $-\bar{A}$. We use construction
to prove the sufficiency. For a matrix $A\in C^{n\times n}$, it can
be written as
\begin{equation}
A=Q^{-1}JQ\thinspace,
\end{equation}
where $J$ is its Jordan canonical form and $Q$ is a reversible matrix.
It's well known that its Jordan canonical form consists of several
Jordan blocks $J(\lambda)$, where
\begin{equation}
J(\lambda)==\left(\begin{array}{cccc}
\lambda & 1\\
 & \ddots & \ddots\\
 &  & \lambda & 1\\
 &  &  & \lambda
\end{array}\right)_{m\times m}\thinspace.
\end{equation}
When $m=1$, the Jordan block $J(\lambda)$ is reduced to $\lambda$.
If $A$ is similar to $-\bar{A}$, then its eigenvalues are symmetric
with imaginary axis, and they are either pure imaginary numbers or
complex number pairs of the form $\lambda=a+bi$ and $-\bar{\lambda}=-a+bi$,
where $a$ and $b$ are real numbers. Accordingly, there are two kinds
of matrix blocks
\begin{equation}
F_{1}=J(ia)=\left(\begin{array}{cccc}
ia & 1\\
 & \ddots & \ddots\\
 &  & ia & 1\\
 &  &  & ia
\end{array}\right)_{m\times m}
\end{equation}
and
\begin{equation}
 F_{2}=\left(\begin{array}{cc}
J(a+bi) & 0\\
0 & J(-a+bi)
\end{array}\right)_{2l\times2l}.
\label{eq:J123}
\end{equation}
The Jordan matrix can now be expressed as $J=Diag(M_{1},M_{2},\cdots,M_{k}),$
where $M_{j}$ is in the form of $F_{1}$ or $F_{2}$. In the following,
we prove that both types of matrix blocks are G-Hamiltonian. For the
first type of matrix block $F_{1}$, if $m$ is odd, the corresponding
Hermitian matrix is
\begin{equation}
G=\left(\begin{array}{ccccc}
 &  &  &  & K_{1}\\
 &  &  & K_{1}\\
 &  & \iddots\\
 & K_{1}\\
1
\end{array}\right)\text{,}\thinspace\text{where}\thinspace K_{1}=\left(\begin{array}{cc}
0 & 1\\
-1 & 0
\end{array}\right)\thinspace.
\end{equation}
For $F_{1}$ with even order and $F_{2}$, the corresponding Hermitian
matrix is
\begin{equation}
G=\left(\begin{array}{cccc}
 &  &  & K_{2}\\
 &  & K_{2}\\
 & \iddots\\
K_{2}
\end{array}\right)\thinspace,\text{where}\thinspace K_{2}=\left(\begin{array}{cc}
0 & i\\
-i & 0
\end{array}\right).\label{eq:2l-1}
\end{equation}
Then for the matrix blocks $M_{j}$, they can be written as $M_{j}=iG_{j}^{-1}S_{j}$.
We construct $G^{'}$ and $S^{'}$ using $G_{j}$ and $S_{j}$ respectively
as follows
\begin{equation}
\begin{aligned}G^{'} & =Diag(G_{1},G_{2},\cdots,G_{k}),\\
S^{'} & =Diag(S_{1},S_{2},\cdots,S_{k}),
\end{aligned}
\end{equation}
and the Jordan canonical form is $J=iG^{'-1}S^{'}$. Then we have
\begin{equation}
A=Q^{-1}JQ=Q^{-1}iG^{'-1}S^{'}Q=iQ^{-1}G^{'-1}Q^{-*}Q^{*}S^{'}Q\thinspace.
\end{equation}
Let
\begin{align}
G & =Q^{*}G^{'}Q\label{eq:Q-1}\\
S & =Q^{*}S^{'}Q\thinspace,\label{eq:S-1}
\end{align}
we obtain
\begin{equation}
A=iG^{-1}S\thinspace,
\end{equation}
where $G$ and $S$ are Hermitian matrix, and $G$ is non-singular.
\end{proof}
The theorem is proved by constructing a non-singular Hermitian matrix
$G$ for a matrix $A$ similar to $-\bar{A}$, but $G$ is not unique.
Usually, we can find more than one non-singular Hermitian matrices
$G$ for a given matrix $A$.

Because for a PT-symmetric Hamiltonian $H$ satisfying Eq.\,\eqref{eq:PT1},
$A$ is indeed similar to $-\bar{A}$, we have the following important
conclusion.
\begin{cor}
For finite dimensional systems, a PT-symmetric Hamiltonian $H$ is
also pseudo-Hermitian. \label{corollary}
\end{cor}

We would like to emphasize again that this fact holds regardless whether
$H$ is diagonalizable or not, and some interesting PT-symmetry breaking
occurs at EPs where $H$ is not diagonalizable.

The fact that a PT-symmetric system is also pseudo-Hermitian or G-Hamiltonian
can be utilized to investigate the mechanism of PT-symmetry breaking.
The dynamics of G-Hamiltonian system has been thoroughly developed
by Krein, Gel'fand and Lidskii \cite{Krein1950,Gel1955,KGML1958}
in 1950s, and the results can be directly applied to PT-symmetric
systems. Specifically, G-Hamiltonian theory gives a comprehensive
description on how a stable system becomes unstable as the system
varies. In terms of the eigenvalues of the Hamiltonian $H$, this
description is about how real eigenvalues of $H$ evolve into conjugate
pairs of complex eigenvalues, in another word, how PT-symmetry breaking
happens. Let's briefly summarize the main results of G-Hamiltonian
theory. (i) The eigenvalues of a G-Hamiltonian matrix $A$ are symmetric
with respect to imaginary axis. They are either pure imaginary numbers
or complex pairs. (ii) Let $\psi$ be an eigenmode (or eigenvector)
of $A$, an action of $\psi$ can be defined as \cite{zhang2016structure}
\[
Ac[\psi]\equiv\left\langle \mathbf{\psi},\psi\right\rangle =\psi^{*}G\psi\thinspace.
\]
Mathematically, this is known as the Krein signature \cite{Krein1950,Gel1955,KGML1958}.
It is called action because it has the dimension of {[}energy{]}$\times${[}time{]}.
(ii) The eigenvalues of $A$ can be classified according to the actions
of the corresponding eigenvectors. An $r$-fold eigenvalue $\lambda$
($Re(\lambda)=0$) of $A$ with its eigen-subspace $V_{\lambda}$
is called the first kind if all eigenmodes of $\lambda$ have positive
actions, i.e., $Ac[\mathbf{y}]>0$ for any $\mathbf{y}\neq0$ in $V_{\lambda}$.
It is called the second kind if all eigenmodes of $\lambda$ have
negative actions. If there exists a zero-action eigenmode, then $\lambda$
is called an eigenvalue of mixed kind \cite{KGML1958}. If an eigenvalue
is the first kind or the second kind, it's called definite. (iii)
The number of each kind of eigenvalue is determined by the Hermitian
matrix $G$. Let p be the number of positive eigenvalues and q be
the number of negative eigenvalues of the matrix $G$, then any G-Hamiltonian
matrix has p eigenvalues of first kind and q eigenvalues of second
kind (counting multiplicity). (iv) The G-Hamiltonian matrix is strongly
stable if and only if all of its eigenvalues lie on the imaginary
axis and are definite. Here, strongly stable means that all eigenvalues
of G-Hamiltonian matrix in an open neighborhood of the parameter space
lie on the imaginary axis.  As a result, a G-Hamiltonian system becomes
unstable when and only when a positive-action mode resonates with
a negative-action mode. This is a process known as the Krein collision.

Applying these results to PT-symmetric systems, we see that PT-symmetry
breaking can happen only when a multiple eigenvalue appears as a result
of two eigenmodes resonate, which occurs at an EP. However, if two
eigenmodes with the same sign of action resonate at an EP, then there
is no PT-symmetry breaking. PT-symmetry breaking is triggered only
when a positive-action mode resonates with a negative-action mode.
In the following, we use some examples to illustrate these facts.

\textbf{Example 1:} Let's study in details the PT-symmetry breaking
for the PT-symmetric Hamiltonian given by Eq.\,\eqref{eq:Ex1}. The
associated coefficient matrix $A=-iH$ is G-Hamiltonian with
\begin{equation}
G=\left(\begin{array}{cccc}
1 & i & 1 & 0\\
-i & 1 & 0 & 1\\
1 & 0 & 0 & i\\
0 & 1 & -i & 0
\end{array}\right)\thinspace.\label{eq:G0}
\end{equation}
This can be directly verified by showing that Eq.\,\eqref{eq:GH}
is satisfied. The eigenvalues of $G$ are $(\left(3\pm\sqrt{5}\right)/2,\left(-1\pm\sqrt{5}\right)/2)$,
three of which are positive and one is negative. Thus, three eigenvalues
of $A$ have positive action (or Krein signatures) and the other one
has negative action. We calculate the eigenmodes of $A$ numerically
and plot the changing process in Fig.\,\ref{f0}, where three eigenmodes
with positive action are marked by $M_{+}$ and the other one with
negative action is marked by $M_{-}$. As shown in Fig.\,\ref{f0}(a),
we fix $a=1$ and $c=-1$, and vary the parameter $b$ from $2$ to
$1/2$. When $b=2$, the eigenvalues of $A$ are all on the imaginary
axis. When $b=a=1$, two eigenmodes with different signs of actions
resonant on the imaginary axis. As $b$ decreasing the two eigenvalues
move off from the imaginary axis and PT-breaking occurs. In Fig.\,\ref{f0}(b),
we plot the changing process of eigenmodes of $A$ by varying $c$
from $-1$ to $1$ and fixing the other parameters $a=1$ and $b=2$.
As $c$ varying, two eigenmodes with opposite signs of actions are
locked on the imaginary axis and the other two with positive action
move towards each other. When $c=0$, the traveling two eigenmodes
collide at the EP. But when increasing $c$ to $1$, they still stay
on the imaginary axis and there is no PT-symmetry breaking. These
figures demonstrate that at the EP $b=a=1$ where two eigenmodes with
different signs of action resonant, PT-symmetry breaking occurs, and
at the EP $c=0$ where the collided eigenmodes have the same sign
of actions, PT-symmetry breaking does not occur. Meanwhile, we find
that at the EPs $b=a=1$, $H$ is not diagonalizable.

On the other hand, we parameterize $a$ and $b$ as $a=1+\sqrt{2}\cos(t)$
and $b=1+\sqrt{2}\sin(t)$, the Hamiltonian becomes

\begin{widetext}
\begin{align}
H^{'}&=\left(\begin{array}{cccc}
-3 & c & 0 & 0\\
c & -3 & 0 & 0\\
1+\sqrt{2}\sin(t)-ci & 7i+1+\sqrt{2}\cos(t) & 4-i(1+\sqrt{2}\cos(t)) & i(1+\sqrt{2}\sin(t))\\
-7i+1+\sqrt{2}\cos(t) & 1+\sqrt{2}\sin(t)+ic & -i(1+\sqrt{2}\sin(t)) & 4+i(1+\sqrt{2}\cos(t))
\end{array}\right).
\end{align}
\end{widetext}
When varying $t$ form $\pi$ to $3\pi/2$, we obtain the changing
process expressed by the red curve in Fig.\,\ref{f0-1}. When $t=5\pi/4$,
two eigenmodes with different actions collide on the axis and PT-breaking
occurs. The Krein collision process is similar to the process given
in Fig.\,\ref{f0}(a). At this EP $b=a=0$, $H$ is diagonalizable,
i.e., the corresponding eigenvalues resonant, but the eigenvectors
don't.

\begin{figure}
\includegraphics[scale=0.55]{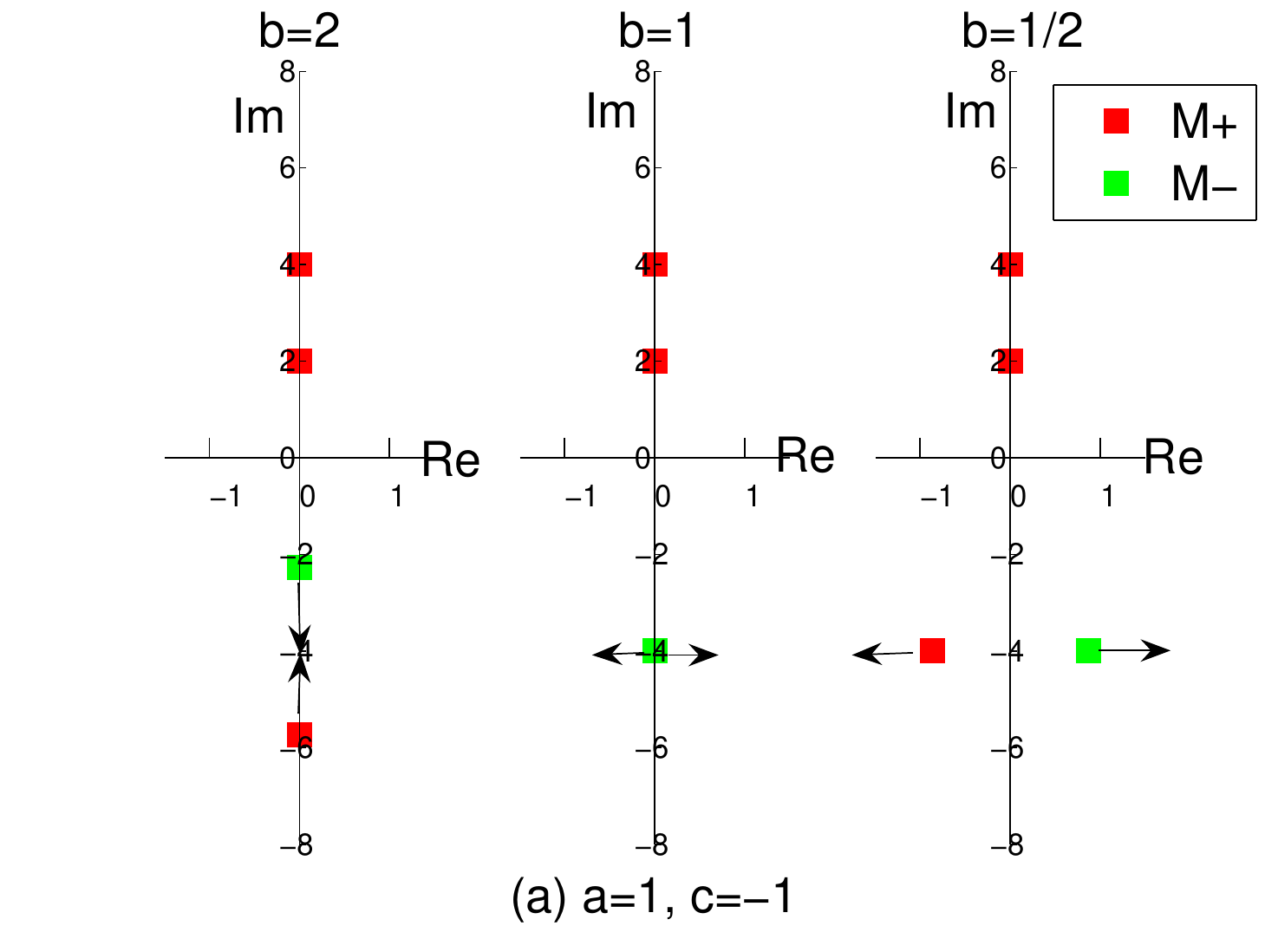}
\includegraphics[scale=0.55]{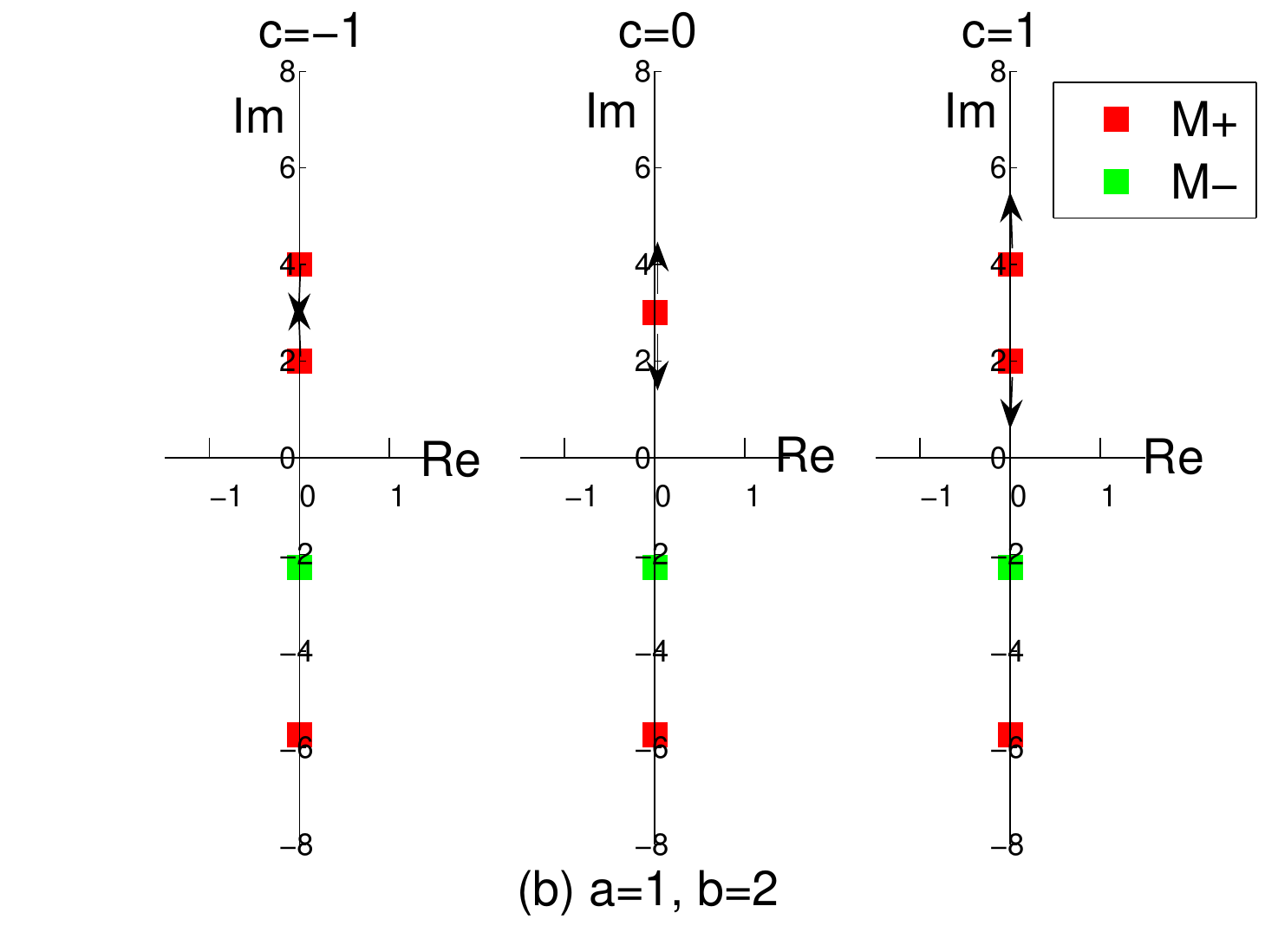}

\caption{(a) PT-symmetry breaking at the EP when a positive-action eigenmode
(red) resonates with a negative eigenmode (green) . (b) No PT-symmetry
breaking at the EP when two positive-action eigenmodes (red) resonate.}
\label{f0}
\end{figure}

\textbf{Example 2:} Consider the system of coupled oscillators
\begin{align}
\ddot{x}+\omega^{2}x+2\gamma\dot{x} & =-\varepsilon y\thinspace,\label{eq:ex21}\\
\ddot{y}+\omega^{2}y-2\gamma\dot{y} & =-\varepsilon x\thinspace,\label{eq:ex22}
\end{align}
where $\omega$, $\varepsilon$ and $\gamma$ are real. This is a
balanced loss-gain system studied in \cite{bender2013twofold}. Similar
and higher dimension examples can be found in refs.\,\cite{schindler2011experimental}
and \cite{bender2014systems}, respectively. In terms of canonical
coordinate $\boldsymbol{x}=(x,y,\dot{x},\dot{y}),$ the system is
\begin{align}
\dot{\boldsymbol{x}} & =A\boldsymbol{x}\label{eq:PT1-1}\\
A & =\left(\begin{array}{cccc}
0 & 0 & 1 & 0\\
0 & 0 & 0 & 1\\
-\omega^{2} & -\varepsilon & -2\gamma & 0\\
-\varepsilon & -\omega^{2} & 0 & 2\gamma
\end{array}\right)\thinspace,\label{eq:PT2-1}
\end{align}
Note that Eq.\,\eqref{eq:PT1-1} is a real non-canonical Hamiltonian
system. The coefficient matrix $A$ is PT-symmetric, i.e., it satisfies
Eq.\,\eqref{eq:PT1} with
\begin{equation}
P=\left(\begin{array}{cccc}
0 & -1 & 0 & 0\\
-1 & 0 & 0 & 0\\
0 & 0 & 0 & 1\\
0 & 0 & 1 & 0
\end{array}\right)\thinspace.
\end{equation}
According to the Corollary \ref{corollary}, $A$ is G-Hamiltonian.
We can verify that the following non-singular Hermitian matrix
\begin{equation}
G=\left(\begin{array}{cccc}
0 & -2i\gamma & 0 & i\\
2i\gamma & 0 & i & 0\\
0 & -i & 0 & 0\\
-i & 0 & 0 & 0
\end{array}\right)\thinspace\label{eq:G1}
\end{equation}
satisfies Eq.\,\eqref{eq:GH}. The eigenvalues of G in Eq.\,\eqref{eq:G1}
are $\lambda=\pm\gamma\pm\sqrt{1+\gamma^{2}}$, two of which are positive
and the other two are negative. Thus two eigenvalues of $A$ have
positive actions (or Krein signatures) and the other two have negative
actions. Let's use numerically calculated examples to observe the
breaking of PT-symmetry through the resonance between a positive-
and a negative-action mode. We plot the process in Fig.\,\ref{f1}
for the case of $\omega=2$ and $\gamma=1$. When $\varepsilon=3.7$,
the eigenvalues of $A$ are all imaginary numbers, two of which have
positive action (marked by $M_{1+}$ and $M_{2+}$) and the other
two have negative action (marked by $M_{1-}$ and $M_{2-}$) in Fig.\,\ref{f1}(a).
Fig.\,\ref{f1}(b) shows that when we decrease $\varepsilon$, $M_{1+}$
and $M_{2-}$ move towards each other, and so do $M_{1-}$ and $M_{2+}$
. Decreasing $\varepsilon$ to $2\sqrt{3}$, eigenmodes $M_{1+}$
and $M_{2-}$ collide on the imaginary axis, and eigenmodes $M_{1-}$
and $M_{2+}$ also collide, as shown in Fig.\,\ref{f1}(c). Because
the resonance are between modes with different sign of actions, the
eigenvalues of $A$ split into pairs symmetric with respect to the
imaginary axis and the PT-symmetry is broken. Fig.\,\ref{f1}(d)
shows that the four eigenvalues of $A$ move out of imaginary axis
when $\varepsilon=3.2$. These figures shows that PT-symmetry breaking
is triggered at EPs where a positive-actions mode resonates with a
negative-action mode. At this EP, $A$ is not diagonalizable.

\begin{figure}
\includegraphics[scale=0.55]{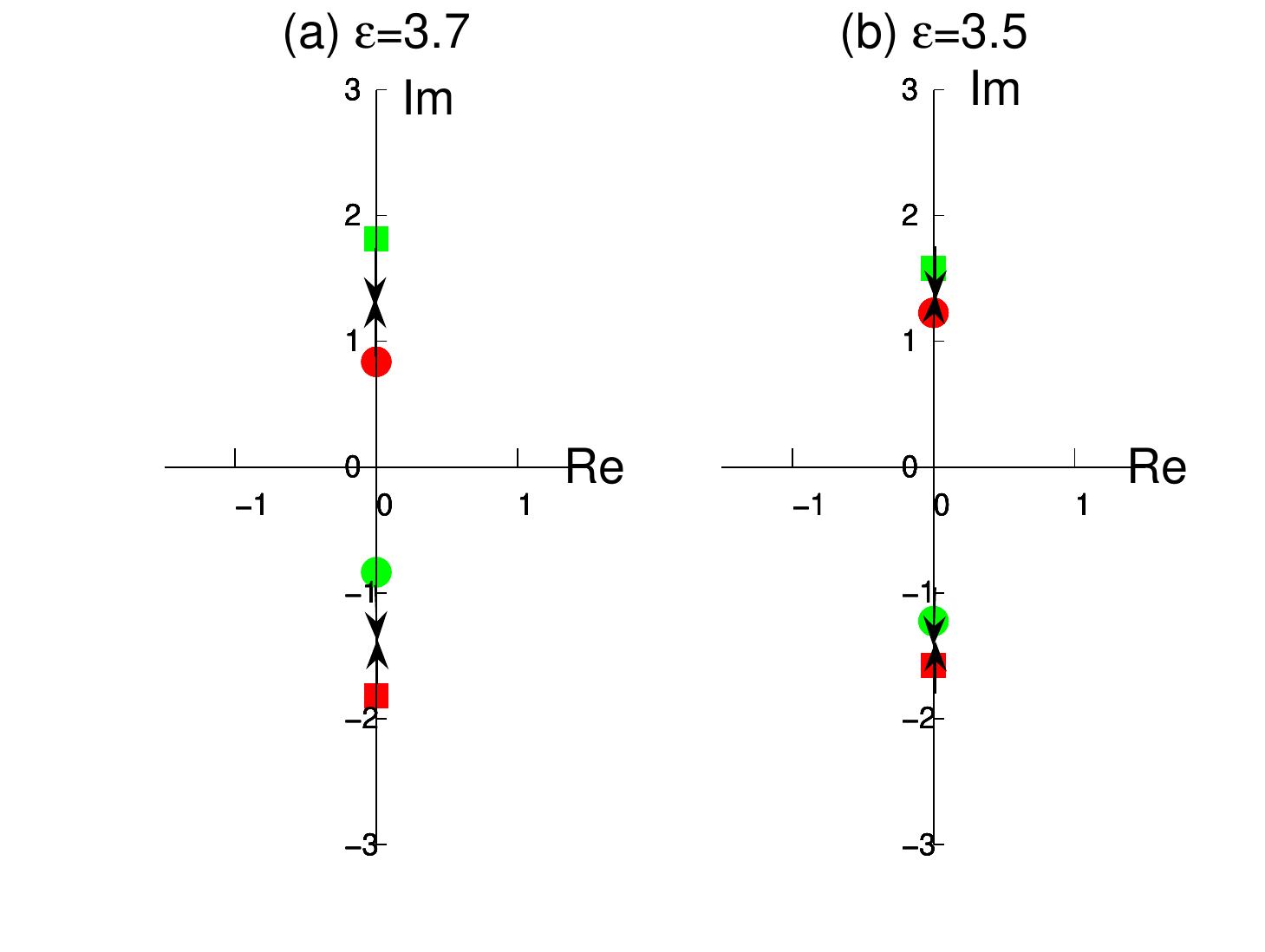}
\includegraphics[scale=0.55]{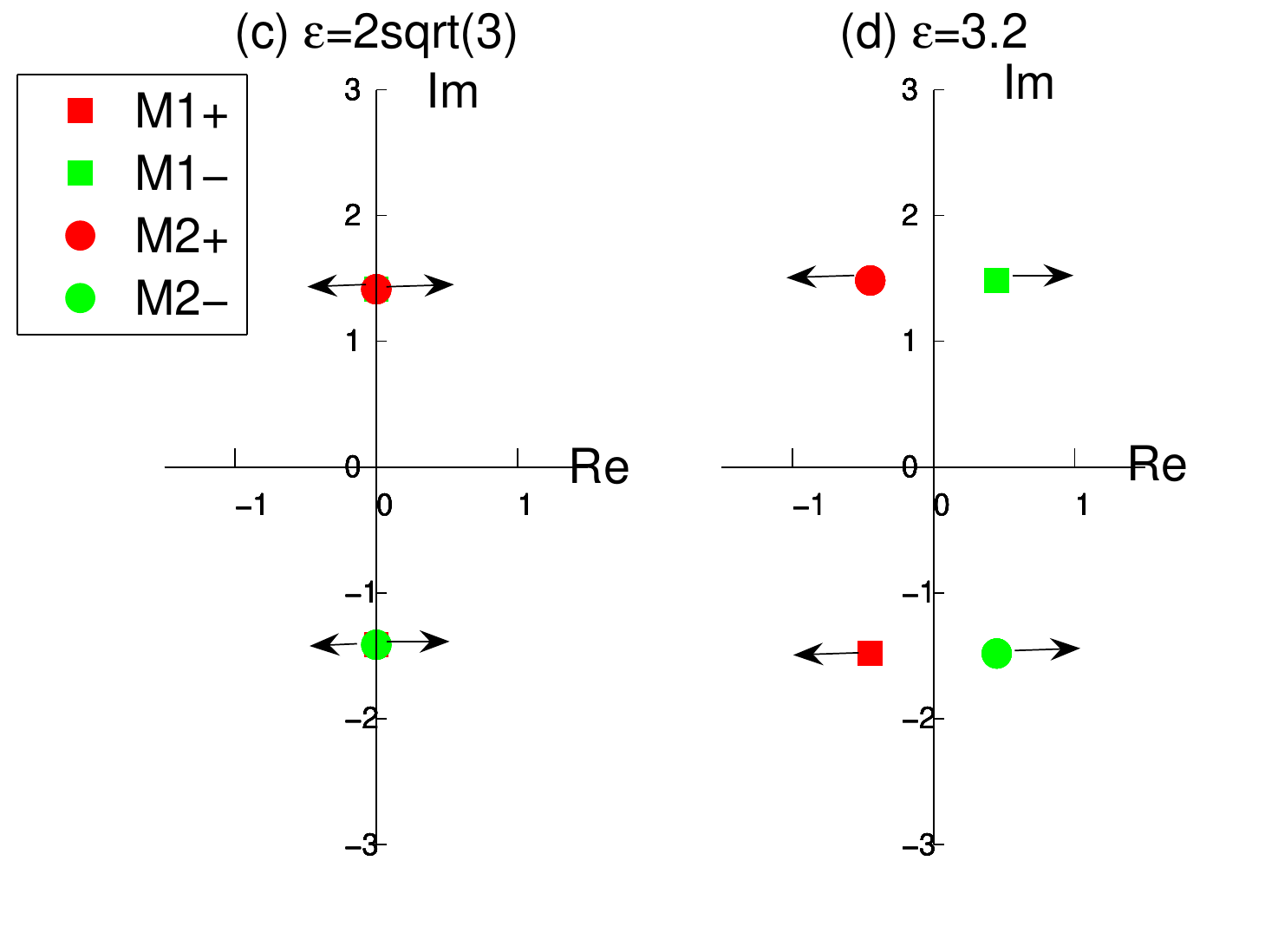}

\caption{PT-symmetry breaking is triggered at the EP where a positive-actions
mode resonates with a negative-action mode. The process is shown for
$\omega=2$, $\gamma=1$ and $\varepsilon$ varying from $3.7$ to
$3.2$. The EP is at $\varepsilon=2\sqrt{3}$, as shown in (c). }
\label{f1}
\end{figure}

In addition to these two examples, we have verified this mechanism
of PT-symmetry breaking for all other PT-symmetric systems that we are aware of.

This research is supported by the National Natural Science Foundation
of China (NSFC-11775219, 11775222, 11505186, 11575185 and 11575186),
the Fundamental Research Funds for the Central Universities (Grant
no. 2017RC033), National Key Research and Development Program (2016YFA0400600,
2016YFA0400601 and 2016YFA0400602), ITER-China Program (2015GB111003,
2014GB124005), the Geo-Algorithmic Plasma Simulator (GAPS) Project,
and the U.S. Department of Energy (DE-AC02-09CH11466).


\begin{thebibliography}{48}%
\makeatletter
\providecommand \@ifxundefined [1]{%
 \@ifx{#1\undefined}
}%
\providecommand \@ifnum [1]{%
 \ifnum #1\expandafter \@firstoftwo
 \else \expandafter \@secondoftwo
 \fi
}%
\providecommand \@ifx [1]{%
 \ifx #1\expandafter \@firstoftwo
 \else \expandafter \@secondoftwo
 \fi
}%
\providecommand \natexlab [1]{#1}%
\providecommand \enquote  [1]{``#1''}%
\providecommand \bibnamefont  [1]{#1}%
\providecommand \bibfnamefont [1]{#1}%
\providecommand \citenamefont [1]{#1}%
\providecommand \href@noop [0]{\@secondoftwo}%
\providecommand \href [0]{\begingroup \@sanitize@url \@href}%
\providecommand \@href[1]{\@@startlink{#1}\@@href}%
\providecommand \@@href[1]{\endgroup#1\@@endlink}%
\providecommand \@sanitize@url [0]{\catcode `\\12\catcode `\$12\catcode
  `\&12\catcode `\#12\catcode `\^12\catcode `\_12\catcode `\%12\relax}%
\providecommand \@@startlink[1]{}%
\providecommand \@@endlink[0]{}%
\providecommand \url  [0]{\begingroup\@sanitize@url \@url }%
\providecommand \@url [1]{\endgroup\@href {#1}{\urlprefix }}%
\providecommand \urlprefix  [0]{URL }%
\providecommand \Eprint [0]{\href }%
\providecommand \doibase [0]{http://dx.doi.org/}%
\providecommand \selectlanguage [0]{\@gobble}%
\providecommand \bibinfo  [0]{\@secondoftwo}%
\providecommand \bibfield  [0]{\@secondoftwo}%
\providecommand \translation [1]{[#1]}%
\providecommand \BibitemOpen [0]{}%
\providecommand \bibitemStop [0]{}%
\providecommand \bibitemNoStop [0]{.\EOS\space}%
\providecommand \EOS [0]{\spacefactor3000\relax}%
\providecommand \BibitemShut  [1]{\csname bibitem#1\endcsname}%
\let\auto@bib@innerbib\@empty
%</preamble>
\bibitem [{\citenamefont {Bender}\ and\ \citenamefont
  {Boettcher}(1998)}]{bender1998real}%
  \BibitemOpen
  \bibfield  {author} {\bibinfo {author} {\bibfnamefont {C.~M.}\ \bibnamefont
  {Bender}}\ and\ \bibinfo {author} {\bibfnamefont {S.}~\bibnamefont
  {Boettcher}},\ }\href@noop {} {\bibfield  {journal} {\bibinfo  {journal}
  {Phys. Rev. Lett.}\ }\textbf {\bibinfo {volume} {80}},\ \bibinfo {pages}
  {5243} (\bibinfo {year} {1998})}\BibitemShut {NoStop}%
\bibitem [{\citenamefont {Bender}\ \emph {et~al.}(2002)\citenamefont {Bender},
  \citenamefont {Brody},\ and\ \citenamefont {Jones}}]{bender2002complex}%
  \BibitemOpen
  \bibfield  {author} {\bibinfo {author} {\bibfnamefont {C.~M.}\ \bibnamefont
  {Bender}}, \bibinfo {author} {\bibfnamefont {D.~C.}\ \bibnamefont {Brody}}, \
  and\ \bibinfo {author} {\bibfnamefont {H.~F.}\ \bibnamefont {Jones}},\
  }\href@noop {} {\bibfield  {journal} {\bibinfo  {journal} {Phys. Rev. Lett.}\
  }\textbf {\bibinfo {volume} {89}},\ \bibinfo {pages} {270401} (\bibinfo
  {year} {2002})}\BibitemShut {NoStop}%
\bibitem [{\citenamefont {Bender}(2007)}]{bender2007making}%
  \BibitemOpen
  \bibfield  {author} {\bibinfo {author} {\bibfnamefont {C.~M.}\ \bibnamefont
  {Bender}},\ }\href@noop {} {\bibfield  {journal} {\bibinfo  {journal} {Rep.
  Prog. Phys.}\ }\textbf {\bibinfo {volume} {70}},\ \bibinfo {pages} {947}
  (\bibinfo {year} {2007})}\BibitemShut {NoStop}%
\bibitem [{\citenamefont {Jones}(1999)}]{jones1999energy}%
  \BibitemOpen
  \bibfield  {author} {\bibinfo {author} {\bibfnamefont {H.}~\bibnamefont
  {Jones}},\ }\href@noop {} {\bibfield  {journal} {\bibinfo  {journal} {Phys.
  Lett. A}\ }\textbf {\bibinfo {volume} {262}},\ \bibinfo {pages} {242}
  (\bibinfo {year} {1999})}\BibitemShut {NoStop}%
\bibitem [{\citenamefont
  {Mostafazadeh}(2002{\natexlab{a}})}]{mostafazadeh2002pseudo}%
  \BibitemOpen
  \bibfield  {author} {\bibinfo {author} {\bibfnamefont {A.}~\bibnamefont
  {Mostafazadeh}},\ }\href@noop {} {\bibfield  {journal} {\bibinfo  {journal}
  {J. Math. Phys.}\ }\textbf {\bibinfo {volume} {43}},\ \bibinfo {pages} {205}
  (\bibinfo {year} {2002}{\natexlab{a}})}\BibitemShut {NoStop}%
\bibitem [{\citenamefont {Heiss}(2004)}]{heiss2004exceptional}%
  \BibitemOpen
  \bibfield  {author} {\bibinfo {author} {\bibfnamefont {W.}~\bibnamefont
  {Heiss}},\ }\href@noop {} {\bibfield  {journal} {\bibinfo  {journal} {J.
  Phys. A: Math. Gen.}\ }\textbf {\bibinfo {volume} {37}},\ \bibinfo {pages}
  {2455} (\bibinfo {year} {2004})}\BibitemShut {NoStop}%
\bibitem [{\citenamefont {El-Ganainy}\ \emph {et~al.}(2007)\citenamefont
  {El-Ganainy}, \citenamefont {Makris}, \citenamefont {Christodoulides},\ and\
  \citenamefont {Musslimani}}]{el2007theory}%
  \BibitemOpen
  \bibfield  {author} {\bibinfo {author} {\bibfnamefont {R.}~\bibnamefont
  {El-Ganainy}}, \bibinfo {author} {\bibfnamefont {K.}~\bibnamefont {Makris}},
  \bibinfo {author} {\bibfnamefont {D.}~\bibnamefont {Christodoulides}}, \ and\
  \bibinfo {author} {\bibfnamefont {Z.~H.}\ \bibnamefont {Musslimani}},\
  }\href@noop {} {\bibfield  {journal} {\bibinfo  {journal} {Opt. Lett.}\
  }\textbf {\bibinfo {volume} {32}},\ \bibinfo {pages} {2632} (\bibinfo {year}
  {2007})}\BibitemShut {NoStop}%
\bibitem [{\citenamefont {Makris}\ \emph {et~al.}(2008)\citenamefont {Makris},
  \citenamefont {El-Ganainy}, \citenamefont {Christodoulides},\ and\
  \citenamefont {Musslimani}}]{makris2008beam}%
  \BibitemOpen
  \bibfield  {author} {\bibinfo {author} {\bibfnamefont {K.~G.}\ \bibnamefont
  {Makris}}, \bibinfo {author} {\bibfnamefont {R.}~\bibnamefont {El-Ganainy}},
  \bibinfo {author} {\bibfnamefont {D.}~\bibnamefont {Christodoulides}}, \ and\
  \bibinfo {author} {\bibfnamefont {Z.~H.}\ \bibnamefont {Musslimani}},\
  }\href@noop {} {\bibfield  {journal} {\bibinfo  {journal} {Phys. Rev. Lett.}\
  }\textbf {\bibinfo {volume} {100}},\ \bibinfo {pages} {103904} (\bibinfo
  {year} {2008})}\BibitemShut {NoStop}%
\bibitem [{\citenamefont {Makris}\ \emph {et~al.}(2010)\citenamefont {Makris},
  \citenamefont {El-Ganainy}, \citenamefont {Christodoulides},\ and\
  \citenamefont {Musslimani}}]{makris2010pt}%
  \BibitemOpen
  \bibfield  {author} {\bibinfo {author} {\bibfnamefont {K.~G.}\ \bibnamefont
  {Makris}}, \bibinfo {author} {\bibfnamefont {R.}~\bibnamefont {El-Ganainy}},
  \bibinfo {author} {\bibfnamefont {D.~N.}\ \bibnamefont {Christodoulides}}, \
  and\ \bibinfo {author} {\bibfnamefont {Z.~H.}\ \bibnamefont {Musslimani}},\
  }\href@noop {} {\bibfield  {journal} {\bibinfo  {journal} {Phys. Rev. A}\
  }\textbf {\bibinfo {volume} {81}},\ \bibinfo {pages} {063807} (\bibinfo
  {year} {2010})}\BibitemShut {NoStop}%
\bibitem [{\citenamefont {Schomerus}(2010)}]{schomerus2010quantum}%
  \BibitemOpen
  \bibfield  {author} {\bibinfo {author} {\bibfnamefont {H.}~\bibnamefont
  {Schomerus}},\ }\href@noop {} {\bibfield  {journal} {\bibinfo  {journal}
  {Phys. Rev. Lett.}\ }\textbf {\bibinfo {volume} {104}},\ \bibinfo {pages}
  {233601} (\bibinfo {year} {2010})}\BibitemShut {NoStop}%
\bibitem [{\citenamefont {Chong}\ \emph {et~al.}(2011)\citenamefont {Chong},
  \citenamefont {Ge},\ and\ \citenamefont {Stone}}]{chong2011p}%
  \BibitemOpen
  \bibfield  {author} {\bibinfo {author} {\bibfnamefont {Y.}~\bibnamefont
  {Chong}}, \bibinfo {author} {\bibfnamefont {L.}~\bibnamefont {Ge}}, \ and\
  \bibinfo {author} {\bibfnamefont {A.~D.}\ \bibnamefont {Stone}},\ }\href@noop
  {} {\bibfield  {journal} {\bibinfo  {journal} {Phys. Rev. Lett.}\ }\textbf
  {\bibinfo {volume} {106}},\ \bibinfo {pages} {093902} (\bibinfo {year}
  {2011})}\BibitemShut {NoStop}%
\bibitem [{\citenamefont {Ge}\ \emph {et~al.}(2012)\citenamefont {Ge},
  \citenamefont {Chong},\ and\ \citenamefont {Stone}}]{ge2012conservation}%
  \BibitemOpen
  \bibfield  {author} {\bibinfo {author} {\bibfnamefont {L.}~\bibnamefont
  {Ge}}, \bibinfo {author} {\bibfnamefont {Y.}~\bibnamefont {Chong}}, \ and\
  \bibinfo {author} {\bibfnamefont {A.~D.}\ \bibnamefont {Stone}},\ }\href@noop
  {} {\bibfield  {journal} {\bibinfo  {journal} {Phys. Rev. A}\ }\textbf
  {\bibinfo {volume} {85}},\ \bibinfo {pages} {023802} (\bibinfo {year}
  {2012})}\BibitemShut {NoStop}%
\bibitem [{\citenamefont {Ramezani}\ \emph {et~al.}(2012)\citenamefont
  {Ramezani}, \citenamefont {Kottos}, \citenamefont {Kovanis},\ and\
  \citenamefont {Christodoulides}}]{ramezani2012exceptional}%
  \BibitemOpen
  \bibfield  {author} {\bibinfo {author} {\bibfnamefont {H.}~\bibnamefont
  {Ramezani}}, \bibinfo {author} {\bibfnamefont {T.}~\bibnamefont {Kottos}},
  \bibinfo {author} {\bibfnamefont {V.}~\bibnamefont {Kovanis}}, \ and\
  \bibinfo {author} {\bibfnamefont {D.~N.}\ \bibnamefont {Christodoulides}},\
  }\href@noop {} {\bibfield  {journal} {\bibinfo  {journal} {Phys. Rev. A}\
  }\textbf {\bibinfo {volume} {85}},\ \bibinfo {pages} {013818} (\bibinfo
  {year} {2012})}\BibitemShut {NoStop}%
\bibitem [{\citenamefont {Klaiman}\ \emph {et~al.}(2008)\citenamefont
  {Klaiman}, \citenamefont {G{\"u}nther},\ and\ \citenamefont
  {Moiseyev}}]{klaiman2008visualization}%
  \BibitemOpen
  \bibfield  {author} {\bibinfo {author} {\bibfnamefont {S.}~\bibnamefont
  {Klaiman}}, \bibinfo {author} {\bibfnamefont {U.}~\bibnamefont
  {G{\"u}nther}}, \ and\ \bibinfo {author} {\bibfnamefont {N.}~\bibnamefont
  {Moiseyev}},\ }\href@noop {} {\bibfield  {journal} {\bibinfo  {journal}
  {Phys. Rev. Lett.}\ }\textbf {\bibinfo {volume} {101}},\ \bibinfo {pages}
  {080402} (\bibinfo {year} {2008})}\BibitemShut {NoStop}%
\bibitem [{\citenamefont {Schindler}\ \emph {et~al.}(2011)\citenamefont
  {Schindler}, \citenamefont {Li}, \citenamefont {Zheng}, \citenamefont
  {Ellis},\ and\ \citenamefont {Kottos}}]{schindler2011experimental}%
  \BibitemOpen
  \bibfield  {author} {\bibinfo {author} {\bibfnamefont {J.}~\bibnamefont
  {Schindler}}, \bibinfo {author} {\bibfnamefont {A.}~\bibnamefont {Li}},
  \bibinfo {author} {\bibfnamefont {M.~C.}\ \bibnamefont {Zheng}}, \bibinfo
  {author} {\bibfnamefont {F.~M.}\ \bibnamefont {Ellis}}, \ and\ \bibinfo
  {author} {\bibfnamefont {T.}~\bibnamefont {Kottos}},\ }\href@noop {}
  {\bibfield  {journal} {\bibinfo  {journal} {Phys. Rev. A}\ }\textbf {\bibinfo
  {volume} {84}},\ \bibinfo {pages} {040101} (\bibinfo {year}
  {2011})}\BibitemShut {NoStop}%
\bibitem [{\citenamefont {Peng}\ \emph {et~al.}(2014)\citenamefont {Peng},
  \citenamefont {{\"O}zdemir}, \citenamefont {Lei}, \citenamefont {Monifi},
  \citenamefont {Gianfreda}, \citenamefont {Long}, \citenamefont {Fan},
  \citenamefont {Nori}, \citenamefont {Bender},\ and\ \citenamefont
  {Yang}}]{peng2014parity}%
  \BibitemOpen
  \bibfield  {author} {\bibinfo {author} {\bibfnamefont {B.}~\bibnamefont
  {Peng}}, \bibinfo {author} {\bibfnamefont {{\c{S}}.~K.}\ \bibnamefont
  {{\"O}zdemir}}, \bibinfo {author} {\bibfnamefont {F.}~\bibnamefont {Lei}},
  \bibinfo {author} {\bibfnamefont {F.}~\bibnamefont {Monifi}}, \bibinfo
  {author} {\bibfnamefont {M.}~\bibnamefont {Gianfreda}}, \bibinfo {author}
  {\bibfnamefont {G.~L.}\ \bibnamefont {Long}}, \bibinfo {author}
  {\bibfnamefont {S.}~\bibnamefont {Fan}}, \bibinfo {author} {\bibfnamefont
  {F.}~\bibnamefont {Nori}}, \bibinfo {author} {\bibfnamefont {C.~M.}\
  \bibnamefont {Bender}}, \ and\ \bibinfo {author} {\bibfnamefont
  {L.}~\bibnamefont {Yang}},\ }\href@noop {} {\bibfield  {journal} {\bibinfo
  {journal} {Nat. Phys.}\ }\textbf {\bibinfo {volume} {10}},\ \bibinfo {pages}
  {394} (\bibinfo {year} {2014})}\BibitemShut {NoStop}%
\bibitem [{\citenamefont {Hodaei}\ \emph {et~al.}(2017)\citenamefont {Hodaei},
  \citenamefont {Hassan}, \citenamefont {Wittek}, \citenamefont
  {Garcia-Gracia}, \citenamefont {El-Ganainy}, \citenamefont
  {Christodoulides},\ and\ \citenamefont {Khajavikhan}}]{hodaei2017enhanced}%
  \BibitemOpen
  \bibfield  {author} {\bibinfo {author} {\bibfnamefont {H.}~\bibnamefont
  {Hodaei}}, \bibinfo {author} {\bibfnamefont {A.~U.}\ \bibnamefont {Hassan}},
  \bibinfo {author} {\bibfnamefont {S.}~\bibnamefont {Wittek}}, \bibinfo
  {author} {\bibfnamefont {H.}~\bibnamefont {Garcia-Gracia}}, \bibinfo {author}
  {\bibfnamefont {R.}~\bibnamefont {El-Ganainy}}, \bibinfo {author}
  {\bibfnamefont {D.~N.}\ \bibnamefont {Christodoulides}}, \ and\ \bibinfo
  {author} {\bibfnamefont {M.}~\bibnamefont {Khajavikhan}},\ }\href@noop {}
  {\bibfield  {journal} {\bibinfo  {journal} {Nature}\ }\textbf {\bibinfo
  {volume} {548}},\ \bibinfo {pages} {187} (\bibinfo {year}
  {2017})}\BibitemShut {NoStop}%
\bibitem [{\citenamefont {Dorey}\ \emph {et~al.}(2007)\citenamefont {Dorey},
  \citenamefont {Dunning},\ and\ \citenamefont {Tateo}}]{dorey2007ode}%
  \BibitemOpen
  \bibfield  {author} {\bibinfo {author} {\bibfnamefont {P.}~\bibnamefont
  {Dorey}}, \bibinfo {author} {\bibfnamefont {C.}~\bibnamefont {Dunning}}, \
  and\ \bibinfo {author} {\bibfnamefont {R.}~\bibnamefont {Tateo}},\
  }\href@noop {} {\bibfield  {journal} {\bibinfo  {journal} {J. Phys. A: Math.
  Theor.}\ }\textbf {\bibinfo {volume} {40}},\ \bibinfo {pages} {R205}
  (\bibinfo {year} {2007})}\BibitemShut {NoStop}%
\bibitem [{\citenamefont {Musslimani}\ \emph {et~al.}(2008)\citenamefont
  {Musslimani}, \citenamefont {Makris}, \citenamefont {El-Ganainy},\ and\
  \citenamefont {Christodoulides}}]{musslimani2008optical}%
  \BibitemOpen
  \bibfield  {author} {\bibinfo {author} {\bibfnamefont {Z.}~\bibnamefont
  {Musslimani}}, \bibinfo {author} {\bibfnamefont {K.~G.}\ \bibnamefont
  {Makris}}, \bibinfo {author} {\bibfnamefont {R.}~\bibnamefont {El-Ganainy}},
  \ and\ \bibinfo {author} {\bibfnamefont {D.~N.}\ \bibnamefont
  {Christodoulides}},\ }\href@noop {} {\bibfield  {journal} {\bibinfo
  {journal} {Phys. Rev. Lett.}\ }\textbf {\bibinfo {volume} {100}},\ \bibinfo
  {pages} {030402} (\bibinfo {year} {2008})}\BibitemShut {NoStop}%
\bibitem [{\citenamefont {Longhi}(2009)}]{longhi2009bloch}%
  \BibitemOpen
  \bibfield  {author} {\bibinfo {author} {\bibfnamefont {S.}~\bibnamefont
  {Longhi}},\ }\href@noop {} {\bibfield  {journal} {\bibinfo  {journal} {Phys.
  Rev. Lett.}\ }\textbf {\bibinfo {volume} {103}},\ \bibinfo {pages} {123601}
  (\bibinfo {year} {2009})}\BibitemShut {NoStop}%
\bibitem [{\citenamefont {Lin}\ \emph {et~al.}(2011)\citenamefont {Lin},
  \citenamefont {Ramezani}, \citenamefont {Eichelkraut}, \citenamefont
  {Kottos}, \citenamefont {Cao},\ and\ \citenamefont
  {Christodoulides}}]{lin2011unidirectional}%
  \BibitemOpen
  \bibfield  {author} {\bibinfo {author} {\bibfnamefont {Z.}~\bibnamefont
  {Lin}}, \bibinfo {author} {\bibfnamefont {H.}~\bibnamefont {Ramezani}},
  \bibinfo {author} {\bibfnamefont {T.}~\bibnamefont {Eichelkraut}}, \bibinfo
  {author} {\bibfnamefont {T.}~\bibnamefont {Kottos}}, \bibinfo {author}
  {\bibfnamefont {H.}~\bibnamefont {Cao}}, \ and\ \bibinfo {author}
  {\bibfnamefont {D.~N.}\ \bibnamefont {Christodoulides}},\ }\href@noop {}
  {\bibfield  {journal} {\bibinfo  {journal} {Phys. Rev. Lett.}\ }\textbf
  {\bibinfo {volume} {106}},\ \bibinfo {pages} {213901} (\bibinfo {year}
  {2011})}\BibitemShut {NoStop}%
\bibitem [{\citenamefont {Szameit}\ \emph {et~al.}(2011)\citenamefont
  {Szameit}, \citenamefont {Rechtsman}, \citenamefont {Bahat-Treidel},\ and\
  \citenamefont {Segev}}]{szameit2011p}%
  \BibitemOpen
  \bibfield  {author} {\bibinfo {author} {\bibfnamefont {A.}~\bibnamefont
  {Szameit}}, \bibinfo {author} {\bibfnamefont {M.~C.}\ \bibnamefont
  {Rechtsman}}, \bibinfo {author} {\bibfnamefont {O.}~\bibnamefont
  {Bahat-Treidel}}, \ and\ \bibinfo {author} {\bibfnamefont {M.}~\bibnamefont
  {Segev}},\ }\href@noop {} {\bibfield  {journal} {\bibinfo  {journal} {Phys.
  Rev. A}\ }\textbf {\bibinfo {volume} {84}},\ \bibinfo {pages} {021806}
  (\bibinfo {year} {2011})}\BibitemShut {NoStop}%
\bibitem [{\citenamefont {Regensburger}\ \emph {et~al.}(2012)\citenamefont
  {Regensburger}, \citenamefont {Bersch}, \citenamefont {Miri}, \citenamefont
  {Onishchukov}, \citenamefont {Christodoulides},\ and\ \citenamefont
  {Peschel}}]{regensburger2012parity}%
  \BibitemOpen
  \bibfield  {author} {\bibinfo {author} {\bibfnamefont {A.}~\bibnamefont
  {Regensburger}}, \bibinfo {author} {\bibfnamefont {C.}~\bibnamefont
  {Bersch}}, \bibinfo {author} {\bibfnamefont {M.-A.}\ \bibnamefont {Miri}},
  \bibinfo {author} {\bibfnamefont {G.}~\bibnamefont {Onishchukov}}, \bibinfo
  {author} {\bibfnamefont {D.~N.}\ \bibnamefont {Christodoulides}}, \ and\
  \bibinfo {author} {\bibfnamefont {U.}~\bibnamefont {Peschel}},\ }\href@noop
  {} {\bibfield  {journal} {\bibinfo  {journal} {Nature}\ }\textbf {\bibinfo
  {volume} {488}},\ \bibinfo {pages} {167} (\bibinfo {year}
  {2012})}\BibitemShut {NoStop}%
\bibitem [{\citenamefont {Sarma}\ \emph {et~al.}(2014)\citenamefont {Sarma},
  \citenamefont {Miri}, \citenamefont {Musslimani},\ and\ \citenamefont
  {Christodoulides}}]{sarma2014continuous}%
  \BibitemOpen
  \bibfield  {author} {\bibinfo {author} {\bibfnamefont {A.~K.}\ \bibnamefont
  {Sarma}}, \bibinfo {author} {\bibfnamefont {M.-A.}\ \bibnamefont {Miri}},
  \bibinfo {author} {\bibfnamefont {Z.~H.}\ \bibnamefont {Musslimani}}, \ and\
  \bibinfo {author} {\bibfnamefont {D.~N.}\ \bibnamefont {Christodoulides}},\
  }\href@noop {} {\bibfield  {journal} {\bibinfo  {journal} {Phys. Rev. E}\
  }\textbf {\bibinfo {volume} {89}},\ \bibinfo {pages} {052918} (\bibinfo
  {year} {2014})}\BibitemShut {NoStop}%
\bibitem [{\citenamefont {Ablowitz}\ and\ \citenamefont
  {Musslimani}(2016)}]{ablowitz2016inverse}%
  \BibitemOpen
  \bibfield  {author} {\bibinfo {author} {\bibfnamefont {M.~J.}\ \bibnamefont
  {Ablowitz}}\ and\ \bibinfo {author} {\bibfnamefont {Z.~H.}\ \bibnamefont
  {Musslimani}},\ }\href@noop {} {\bibfield  {journal} {\bibinfo  {journal}
  {Nonlinearity}\ }\textbf {\bibinfo {volume} {29}},\ \bibinfo {pages} {915}
  (\bibinfo {year} {2016})}\BibitemShut {NoStop}%
\bibitem [{\citenamefont {Zhang}\ \emph
  {et~al.}(2016{\natexlab{a}})\citenamefont {Zhang}, \citenamefont {Zhang},
  \citenamefont {Sheng}, \citenamefont {Yang}, \citenamefont {Miri},
  \citenamefont {Christodoulides}, \citenamefont {He}, \citenamefont {Zhang},\
  and\ \citenamefont {Xiao}}]{zhang2016observation}%
  \BibitemOpen
  \bibfield  {author} {\bibinfo {author} {\bibfnamefont {Z.}~\bibnamefont
  {Zhang}}, \bibinfo {author} {\bibfnamefont {Y.}~\bibnamefont {Zhang}},
  \bibinfo {author} {\bibfnamefont {J.}~\bibnamefont {Sheng}}, \bibinfo
  {author} {\bibfnamefont {L.}~\bibnamefont {Yang}}, \bibinfo {author}
  {\bibfnamefont {M.-A.}\ \bibnamefont {Miri}}, \bibinfo {author}
  {\bibfnamefont {D.~N.}\ \bibnamefont {Christodoulides}}, \bibinfo {author}
  {\bibfnamefont {B.}~\bibnamefont {He}}, \bibinfo {author} {\bibfnamefont
  {Y.}~\bibnamefont {Zhang}}, \ and\ \bibinfo {author} {\bibfnamefont
  {M.}~\bibnamefont {Xiao}},\ }\href@noop {} {\bibfield  {journal} {\bibinfo
  {journal} {Phys. Rev. Lett.}\ }\textbf {\bibinfo {volume} {117}},\ \bibinfo
  {pages} {123601} (\bibinfo {year} {2016}{\natexlab{a}})}\BibitemShut
  {NoStop}%
\bibitem [{\citenamefont {Jahromi}\ \emph {et~al.}(2017)\citenamefont
  {Jahromi}, \citenamefont {Hassan}, \citenamefont {Christodoulides},\ and\
  \citenamefont {Abouraddy}}]{jahromi2017statistical}%
  \BibitemOpen
  \bibfield  {author} {\bibinfo {author} {\bibfnamefont {A.~K.}\ \bibnamefont
  {Jahromi}}, \bibinfo {author} {\bibfnamefont {A.~U.}\ \bibnamefont {Hassan}},
  \bibinfo {author} {\bibfnamefont {D.~N.}\ \bibnamefont {Christodoulides}}, \
  and\ \bibinfo {author} {\bibfnamefont {A.~F.}\ \bibnamefont {Abouraddy}},\
  }\href@noop {} {\bibfield  {journal} {\bibinfo  {journal} {Nat. Commun.}\
  }\textbf {\bibinfo {volume} {8}},\ \bibinfo {pages} {1359} (\bibinfo {year}
  {2017})}\BibitemShut {NoStop}%
\bibitem [{\citenamefont {Guo}\ \emph {et~al.}(2009)\citenamefont {Guo},
  \citenamefont {Salamo}, \citenamefont {Duchesne}, \citenamefont {Morandotti},
  \citenamefont {Volatier-Ravat}, \citenamefont {Aimez}, \citenamefont
  {Siviloglou},\ and\ \citenamefont {Christodoulides}}]{guo2009observation}%
  \BibitemOpen
  \bibfield  {author} {\bibinfo {author} {\bibfnamefont {A.}~\bibnamefont
  {Guo}}, \bibinfo {author} {\bibfnamefont {G.}~\bibnamefont {Salamo}},
  \bibinfo {author} {\bibfnamefont {D.}~\bibnamefont {Duchesne}}, \bibinfo
  {author} {\bibfnamefont {R.}~\bibnamefont {Morandotti}}, \bibinfo {author}
  {\bibfnamefont {M.}~\bibnamefont {Volatier-Ravat}}, \bibinfo {author}
  {\bibfnamefont {V.}~\bibnamefont {Aimez}}, \bibinfo {author} {\bibfnamefont
  {G.}~\bibnamefont {Siviloglou}}, \ and\ \bibinfo {author} {\bibfnamefont
  {D.}~\bibnamefont {Christodoulides}},\ }\href@noop {} {\bibfield  {journal}
  {\bibinfo  {journal} {Phys. Rev. Lett.}\ }\textbf {\bibinfo {volume} {103}},\
  \bibinfo {pages} {093902} (\bibinfo {year} {2009})}\BibitemShut {NoStop}%
\bibitem [{\citenamefont {R{\"u}ter}\ \emph {et~al.}(2010)\citenamefont
  {R{\"u}ter}, \citenamefont {Makris}, \citenamefont {El-Ganainy},
  \citenamefont {Christodoulides}, \citenamefont {Segev},\ and\ \citenamefont
  {Kip}}]{ruter2010observation}%
  \BibitemOpen
  \bibfield  {author} {\bibinfo {author} {\bibfnamefont {C.~E.}\ \bibnamefont
  {R{\"u}ter}}, \bibinfo {author} {\bibfnamefont {K.~G.}\ \bibnamefont
  {Makris}}, \bibinfo {author} {\bibfnamefont {R.}~\bibnamefont {El-Ganainy}},
  \bibinfo {author} {\bibfnamefont {D.~N.}\ \bibnamefont {Christodoulides}},
  \bibinfo {author} {\bibfnamefont {M.}~\bibnamefont {Segev}}, \ and\ \bibinfo
  {author} {\bibfnamefont {D.}~\bibnamefont {Kip}},\ }\href@noop {} {\bibfield
  {journal} {\bibinfo  {journal} {Nat. Phys.}\ }\textbf {\bibinfo {volume}
  {6}},\ \bibinfo {pages} {192} (\bibinfo {year} {2010})}\BibitemShut {NoStop}%
\bibitem [{\citenamefont {Feng}\ \emph {et~al.}(2011)\citenamefont {Feng},
  \citenamefont {Ayache}, \citenamefont {Huang}, \citenamefont {Xu},
  \citenamefont {Lu}, \citenamefont {Chen}, \citenamefont {Fainman},\ and\
  \citenamefont {Scherer}}]{feng2011nonreciprocal}%
  \BibitemOpen
  \bibfield  {author} {\bibinfo {author} {\bibfnamefont {L.}~\bibnamefont
  {Feng}}, \bibinfo {author} {\bibfnamefont {M.}~\bibnamefont {Ayache}},
  \bibinfo {author} {\bibfnamefont {J.}~\bibnamefont {Huang}}, \bibinfo
  {author} {\bibfnamefont {Y.-L.}\ \bibnamefont {Xu}}, \bibinfo {author}
  {\bibfnamefont {M.-H.}\ \bibnamefont {Lu}}, \bibinfo {author} {\bibfnamefont
  {Y.-F.}\ \bibnamefont {Chen}}, \bibinfo {author} {\bibfnamefont
  {Y.}~\bibnamefont {Fainman}}, \ and\ \bibinfo {author} {\bibfnamefont
  {A.}~\bibnamefont {Scherer}},\ }\href@noop {} {\bibfield  {journal} {\bibinfo
   {journal} {Science}\ }\textbf {\bibinfo {volume} {333}},\ \bibinfo {pages}
  {729} (\bibinfo {year} {2011})}\BibitemShut {NoStop}%
\bibitem [{\citenamefont {Bittner}\ \emph {et~al.}(2012)\citenamefont
  {Bittner}, \citenamefont {Dietz}, \citenamefont {G{\"u}nther}, \citenamefont
  {Harney}, \citenamefont {Miski-Oglu}, \citenamefont {Richter},\ and\
  \citenamefont {Sch{\"a}fer}}]{bittner2012p}%
  \BibitemOpen
  \bibfield  {author} {\bibinfo {author} {\bibfnamefont {S.}~\bibnamefont
  {Bittner}}, \bibinfo {author} {\bibfnamefont {B.}~\bibnamefont {Dietz}},
  \bibinfo {author} {\bibfnamefont {U.}~\bibnamefont {G{\"u}nther}}, \bibinfo
  {author} {\bibfnamefont {H.}~\bibnamefont {Harney}}, \bibinfo {author}
  {\bibfnamefont {M.}~\bibnamefont {Miski-Oglu}}, \bibinfo {author}
  {\bibfnamefont {A.}~\bibnamefont {Richter}}, \ and\ \bibinfo {author}
  {\bibfnamefont {F.}~\bibnamefont {Sch{\"a}fer}},\ }\href@noop {} {\bibfield
  {journal} {\bibinfo  {journal} {Phys. Rev. Lett.}\ }\textbf {\bibinfo
  {volume} {108}},\ \bibinfo {pages} {024101} (\bibinfo {year}
  {2012})}\BibitemShut {NoStop}%
\bibitem [{\citenamefont {Liertzer}\ \emph {et~al.}(2012)\citenamefont
  {Liertzer}, \citenamefont {Ge}, \citenamefont {Cerjan}, \citenamefont
  {Stone}, \citenamefont {T{\"u}reci},\ and\ \citenamefont
  {Rotter}}]{liertzer2012pump}%
  \BibitemOpen
  \bibfield  {author} {\bibinfo {author} {\bibfnamefont {M.}~\bibnamefont
  {Liertzer}}, \bibinfo {author} {\bibfnamefont {L.}~\bibnamefont {Ge}},
  \bibinfo {author} {\bibfnamefont {A.}~\bibnamefont {Cerjan}}, \bibinfo
  {author} {\bibfnamefont {A.}~\bibnamefont {Stone}}, \bibinfo {author}
  {\bibfnamefont {H.}~\bibnamefont {T{\"u}reci}}, \ and\ \bibinfo {author}
  {\bibfnamefont {S.}~\bibnamefont {Rotter}},\ }\href@noop {} {\bibfield
  {journal} {\bibinfo  {journal} {Phys. Rev. Lett.}\ }\textbf {\bibinfo
  {volume} {108}},\ \bibinfo {pages} {173901} (\bibinfo {year}
  {2012})}\BibitemShut {NoStop}%
\bibitem [{\citenamefont {Kato}(1966)}]{kato1966perturbation}%
  \BibitemOpen
  \bibfield  {author} {\bibinfo {author} {\bibfnamefont {T.}~\bibnamefont
  {Kato}},\ }\enquote {\bibinfo {title} {Perturbation theory for linear
  operators},}\ \ (\bibinfo  {publisher} {Springer-Verlag, Berlin},\ \bibinfo
  {year} {1966})\ p.~\bibinfo {pages} {64}\BibitemShut {NoStop}%
\bibitem [{\citenamefont {Brandstetter}\ \emph {et~al.}(2014)\citenamefont
  {Brandstetter}, \citenamefont {Liertzer}, \citenamefont {Deutsch},
  \citenamefont {Klang}, \citenamefont {Sch{\"o}berl}, \citenamefont
  {T{\"u}reci}, \citenamefont {Strasser}, \citenamefont {Unterrainer},\ and\
  \citenamefont {Rotter}}]{brandstetter2014reversing}%
  \BibitemOpen
  \bibfield  {author} {\bibinfo {author} {\bibfnamefont {M.}~\bibnamefont
  {Brandstetter}}, \bibinfo {author} {\bibfnamefont {M.}~\bibnamefont
  {Liertzer}}, \bibinfo {author} {\bibfnamefont {C.}~\bibnamefont {Deutsch}},
  \bibinfo {author} {\bibfnamefont {P.}~\bibnamefont {Klang}}, \bibinfo
  {author} {\bibfnamefont {J.}~\bibnamefont {Sch{\"o}berl}}, \bibinfo {author}
  {\bibfnamefont {H.}~\bibnamefont {T{\"u}reci}}, \bibinfo {author}
  {\bibfnamefont {G.}~\bibnamefont {Strasser}}, \bibinfo {author}
  {\bibfnamefont {K.}~\bibnamefont {Unterrainer}}, \ and\ \bibinfo {author}
  {\bibfnamefont {S.}~\bibnamefont {Rotter}},\ }\href@noop {} {\bibfield
  {journal} {\bibinfo  {journal} {Nat. Commun.}\ }\textbf {\bibinfo {volume}
  {5}},\ \bibinfo {pages} {4034} (\bibinfo {year} {2014})}\BibitemShut
  {NoStop}%
\bibitem [{\citenamefont {Heiss}(2012)}]{heiss2012physics}%
  \BibitemOpen
  \bibfield  {author} {\bibinfo {author} {\bibfnamefont {W.}~\bibnamefont
  {Heiss}},\ }\href@noop {} {\bibfield  {journal} {\bibinfo  {journal} {J.
  Phys. A: Math. Theor.}\ }\textbf {\bibinfo {volume} {45}},\ \bibinfo {pages}
  {444016} (\bibinfo {year} {2012})}\BibitemShut {NoStop}%
\bibitem [{\citenamefont {Ge}(2016)}]{ge2016anomalous}%
  \BibitemOpen
  \bibfield  {author} {\bibinfo {author} {\bibfnamefont {L.}~\bibnamefont
  {Ge}},\ }\href@noop {} {\bibfield  {journal} {\bibinfo  {journal} {Phys. Rev.
  A}\ }\textbf {\bibinfo {volume} {94}},\ \bibinfo {pages} {013837} (\bibinfo
  {year} {2016})}\BibitemShut {NoStop}%
\bibitem [{\citenamefont {Feng}\ \emph {et~al.}(2017)\citenamefont {Feng},
  \citenamefont {El-Ganainy},\ and\ \citenamefont {Ge}}]{feng2017non}%
  \BibitemOpen
  \bibfield  {author} {\bibinfo {author} {\bibfnamefont {L.}~\bibnamefont
  {Feng}}, \bibinfo {author} {\bibfnamefont {R.}~\bibnamefont {El-Ganainy}}, \
  and\ \bibinfo {author} {\bibfnamefont {L.}~\bibnamefont {Ge}},\ }\href@noop
  {} {\bibfield  {journal} {\bibinfo  {journal} {Nat. Photonics}\ }\textbf
  {\bibinfo {volume} {11}},\ \bibinfo {pages} {752} (\bibinfo {year}
  {2017})}\BibitemShut {NoStop}%
\bibitem [{\citenamefont {Ashida}\ \emph {et~al.}(2017)\citenamefont {Ashida},
  \citenamefont {Furukawa},\ and\ \citenamefont {Ueda}}]{Ashida2017PT}%
  \BibitemOpen
  \bibfield  {author} {\bibinfo {author} {\bibfnamefont {Y.}~\bibnamefont
  {Ashida}}, \bibinfo {author} {\bibfnamefont {S.}~\bibnamefont {Furukawa}}, \
  and\ \bibinfo {author} {\bibfnamefont {M.}~\bibnamefont {Ueda}},\ }\href@noop
  {} {\bibfield  {journal} {\bibinfo  {journal} {Nat. Commun.}\ }\textbf
  {\bibinfo {volume} {8}},\ \bibinfo {pages} {15791} (\bibinfo {year}
  {2017})}\BibitemShut {NoStop}%
\bibitem [{\citenamefont {El-Ganainy}\ \emph {et~al.}(2018)\citenamefont
  {El-Ganainy}, \citenamefont {Makris}, \citenamefont {Khajavikhan},
  \citenamefont {Musslimani}, \citenamefont {Rotter},\ and\ \citenamefont
  {Christodoulides}}]{el2018non}%
  \BibitemOpen
  \bibfield  {author} {\bibinfo {author} {\bibfnamefont {R.}~\bibnamefont
  {El-Ganainy}}, \bibinfo {author} {\bibfnamefont {K.~G.}\ \bibnamefont
  {Makris}}, \bibinfo {author} {\bibfnamefont {M.}~\bibnamefont {Khajavikhan}},
  \bibinfo {author} {\bibfnamefont {Z.~H.}\ \bibnamefont {Musslimani}},
  \bibinfo {author} {\bibfnamefont {S.}~\bibnamefont {Rotter}}, \ and\ \bibinfo
  {author} {\bibfnamefont {D.~N.}\ \bibnamefont {Christodoulides}},\
  }\href@noop {} {\bibfield  {journal} {\bibinfo  {journal} {Nat. Phys.}\
  }\textbf {\bibinfo {volume} {14}},\ \bibinfo {pages} {11} (\bibinfo {year}
  {2018})}\BibitemShut {NoStop}%
\bibitem [{\citenamefont {Lee}\ and\ \citenamefont {Wich}(1969)}]{Lee69}%
  \BibitemOpen
  \bibfield  {author} {\bibinfo {author} {\bibfnamefont {T.~D.}\ \bibnamefont
  {Lee}}\ and\ \bibinfo {author} {\bibfnamefont {G.~C.}\ \bibnamefont {Wich}},\
  }\href@noop {} {\bibfield  {journal} {\bibinfo  {journal} {Nucl. Phys. B}\
  }\textbf {\bibinfo {volume} {9}},\ \bibinfo {pages} {209} (\bibinfo {year}
  {1969})}\BibitemShut {NoStop}%
\bibitem [{\citenamefont
  {Mostafazadeh}(2002{\natexlab{b}})}]{mostafazadeh2002pseudoII}%
  \BibitemOpen
  \bibfield  {author} {\bibinfo {author} {\bibfnamefont {A.}~\bibnamefont
  {Mostafazadeh}},\ }\href@noop {} {\bibfield  {journal} {\bibinfo  {journal}
  {J. Math. Phys.}\ }\textbf {\bibinfo {volume} {43}},\ \bibinfo {pages} {2814}
  (\bibinfo {year} {2002}{\natexlab{b}})}\BibitemShut {NoStop}%
\bibitem [{\citenamefont
  {Mostafazadeh}(2002{\natexlab{c}})}]{mostafazadeh2002pseudoIII}%
  \BibitemOpen
  \bibfield  {author} {\bibinfo {author} {\bibfnamefont {A.}~\bibnamefont
  {Mostafazadeh}},\ }\href@noop {} {\bibfield  {journal} {\bibinfo  {journal}
  {J. Math. Phys.}\ }\textbf {\bibinfo {volume} {43}},\ \bibinfo {pages} {3944}
  (\bibinfo {year} {2002}{\natexlab{c}})}\BibitemShut {NoStop}%
\bibitem [{\citenamefont {Krein}(1950)}]{Krein1950}%
  \BibitemOpen
  \bibfield  {author} {\bibinfo {author} {\bibfnamefont {M.}~\bibnamefont
  {Krein}},\ }\href@noop {} {\bibfield  {journal} {\bibinfo  {journal} {Doklady
  Akad. Nauk. SSSR N.S.}\ }\textbf {\bibinfo {volume} {73}},\ \bibinfo {pages}
  {445} (\bibinfo {year} {1950})}\BibitemShut {NoStop}%
\bibitem [{\citenamefont {Gel'fand}\ and\ \citenamefont
  {Lidskii}(1955)}]{Gel1955}%
  \BibitemOpen
  \bibfield  {author} {\bibinfo {author} {\bibfnamefont {I.~M.}\ \bibnamefont
  {Gel'fand}}\ and\ \bibinfo {author} {\bibfnamefont {V.~B.}\ \bibnamefont
  {Lidskii}},\ }\href@noop {} {\bibfield  {journal} {\bibinfo  {journal}
  {Uspekhi Mat. Nauk}\ }\textbf {\bibinfo {volume} {10}},\ \bibinfo {pages} {3}
  (\bibinfo {year} {1955})}\BibitemShut {NoStop}%
\bibitem [{\citenamefont {Yakubovich}\ and\ \citenamefont
  {Starzhinskii}(1975)}]{KGML1958}%
  \BibitemOpen
  \bibfield  {author} {\bibinfo {author} {\bibfnamefont {V.}~\bibnamefont
  {Yakubovich}}\ and\ \bibinfo {author} {\bibfnamefont {V.}~\bibnamefont
  {Starzhinskii}},\ }\href@noop {} {\emph {\bibinfo {title} {Linear
  Differential Equations with Periodic Coefficients}}},\ Vol.~\bibinfo {volume}
  {I}\ (\bibinfo  {publisher} {Wiley},\ \bibinfo {year} {1975})\BibitemShut
  {NoStop}%
\bibitem [{\citenamefont {Zhang}\ \emph
  {et~al.}(2016{\natexlab{b}})\citenamefont {Zhang}, \citenamefont {Qin},
  \citenamefont {Davidson}, \citenamefont {Liu},\ and\ \citenamefont
  {Xiao}}]{zhang2016structure}%
  \BibitemOpen
  \bibfield  {author} {\bibinfo {author} {\bibfnamefont {R.}~\bibnamefont
  {Zhang}}, \bibinfo {author} {\bibfnamefont {H.}~\bibnamefont {Qin}}, \bibinfo
  {author} {\bibfnamefont {R.~C.}\ \bibnamefont {Davidson}}, \bibinfo {author}
  {\bibfnamefont {J.}~\bibnamefont {Liu}}, \ and\ \bibinfo {author}
  {\bibfnamefont {J.}~\bibnamefont {Xiao}},\ }\href@noop {} {\bibfield
  {journal} {\bibinfo  {journal} {Phys. Plasmas}\ }\textbf {\bibinfo {volume}
  {23}},\ \bibinfo {pages} {072111} (\bibinfo {year}
  {2016}{\natexlab{b}})}\BibitemShut {NoStop}%
\bibitem [{\citenamefont {Bender}\ \emph {et~al.}(2013)\citenamefont {Bender},
  \citenamefont {Gianfreda}, \citenamefont {{\"O}zdemir}, \citenamefont
  {Peng},\ and\ \citenamefont {Yang}}]{bender2013twofold}%
  \BibitemOpen
  \bibfield  {author} {\bibinfo {author} {\bibfnamefont {C.~M.}\ \bibnamefont
  {Bender}}, \bibinfo {author} {\bibfnamefont {M.}~\bibnamefont {Gianfreda}},
  \bibinfo {author} {\bibfnamefont {{\c{S}}.~K.}\ \bibnamefont {{\"O}zdemir}},
  \bibinfo {author} {\bibfnamefont {B.}~\bibnamefont {Peng}}, \ and\ \bibinfo
  {author} {\bibfnamefont {L.}~\bibnamefont {Yang}},\ }\href@noop {} {\bibfield
   {journal} {\bibinfo  {journal} {Phys. Rev. A}\ }\textbf {\bibinfo {volume}
  {88}},\ \bibinfo {pages} {062111} (\bibinfo {year} {2013})}\BibitemShut
  {NoStop}%
\bibitem [{\citenamefont {Bender}\ \emph {et~al.}(2014)\citenamefont {Bender},
  \citenamefont {Gianfreda},\ and\ \citenamefont
  {Klevansky}}]{bender2014systems}%
  \BibitemOpen
  \bibfield  {author} {\bibinfo {author} {\bibfnamefont {C.~M.}\ \bibnamefont
  {Bender}}, \bibinfo {author} {\bibfnamefont {M.}~\bibnamefont {Gianfreda}}, \
  and\ \bibinfo {author} {\bibfnamefont {S.}~\bibnamefont {Klevansky}},\
  }\href@noop {} {\bibfield  {journal} {\bibinfo  {journal} {Phys. Rev. A}\
  }\textbf {\bibinfo {volume} {90}},\ \bibinfo {pages} {022114} (\bibinfo
  {year} {2014})}\BibitemShut {NoStop}%
\end{thebibliography}
\end{document}